\documentclass[conference]{IEEEtran}
\IEEEoverridecommandlockouts

\usepackage[usenames,dvipsnames]{xcolor}
\usepackage{amssymb}
\usepackage{amsthm}
\usepackage{amsmath}
\usepackage{xspace,paralist}
\usepackage[noend]{algpseudocode}
\usepackage{algorithm}
\usepackage{graphicx}
\usepackage{tikz}
\usepackage[sort]{cite}

\usepackage{caption}
\usepackage{subcaption}

\usetikzlibrary{shapes}
\usetikzlibrary{chains,scopes}
\usetikzlibrary{positioning,calc,arrows.meta}
\usepackage{pgfplots, pgfplotstable}
\pgfplotsset{compat=1.17}
\usepackage{scalerel}
\usetikzlibrary{svg.path}

\definecolor{orcidlogocol}{HTML}{A6CE39}
\tikzset{
    orcidlogo/.pic={
        \fill[orcidlogocol] svg{M256,128c0,70.7-57.3,128-128,128C57.3,256,0,198.7,0,128C0,57.3,57.3,0,128,0C198.7,0,256,57.3,256,128z};
        \fill[white] svg{M86.3,186.2H70.9V79.1h15.4v48.4V186.2z}
        svg{M108.9,79.1h41.6c39.6,0,57,28.3,57,53.6c0,27.5-21.5,53.6-56.8,53.6h-41.8V79.1z M124.3,172.4h24.5c34.9,0,42.9-26.5,42.9-39.7c0-21.5-13.7-39.7-43.7-39.7h-23.7V172.4z}
        svg{M88.7,56.8c0,5.5-4.5,10.1-10.1,10.1c-5.6,0-10.1-4.6-10.1-10.1c0-5.6,4.5-10.1,10.1-10.1C84.2,46.7,88.7,51.3,88.7,56.8z};
    }
}

\newcommand\orcidicon[1]{\href{https://orcid.org/#1}{\mbox{\scalerel*{
                \begin{tikzpicture}[yscale=-1,transform shape]
                \pic{orcidlogo};
                \end{tikzpicture}
            }{|}}}}
\usepackage{hyperref}
\usepackage[capitalise,nameinlink,noabbrev]{cleveref}
\newif\ifextendedversion
\extendedversionfalse
\ifextendedversion
\pratendSetGlobal{text link={See \hyperref[proof:prAtEnd\thecounterAllProofEnd]{proof} in the extended version~\cite{extendeddoacc} }}
\fi
\newif\ifasyncdstonode
\asyncdstonodefalse
\newtheorem{theorem}{Theorem}
\newcommand{\agreedist}{\ensuremath{\mathfrak{d}}\xspace}

\newcommand{\dstimebound}{\ensuremath{T_{ds}}\xspace}

\newcommand{\true}{\ensuremath{\mathit{true}}\xspace}

\newcommand{\smr}{\ensuremath{\mathsf{SMR}}\xspace}
\newcommand{\command}[1]{\ensuremath{\mathsf{#1}}\xspace}

\newcommand{\doracc}{DORA-CC\xspace}
\newcommand{\dora}{DORA\xspace}

\newcommand{\sotaoracle}{OCR\xspace}
\newcommand{\fallbacktimer}{\ensuremath{T_{fallback}}\xspace}
\newcommand{\commodity}{\ensuremath{\tau}\xspace}
\newcommand{\svalue}{\ensuremath{\mathcal{S}}\xspace}
\newcommand{\BTC}{\ensuremath{\mathtt{BTC}}\xspace}

\newcommand{\USD}{\ensuremath{\mathtt{USD}}\xspace}

\newcommand{\Nat}{\ensuremath{\mathbb{N}}}

\newcommand{\median}[1]{\ensuremath{\mathcal{Q}_{0.5}(#1)}}

\newcommand{\Omit}[1]{}
\newcommand{\agree}[2]{\ensuremath{\|#1-#2\|_1} \leq \agreedist}

\newcommand{\honest}{\ensuremath{\mathcal{H}}}
\newcommand{\hmax}{\ensuremath{\honest_{max}}}
\newcommand{\hmin}{\ensuremath{\honest_{min}}}
\newcommand{\aggset}{\ensuremath{\mathbb{A}}\xspace}
\newcommand{\Byz}{\ensuremath{\mathcal{B}}}
\newcommand{\hldiff}[1]{#1}

\newcommand{\MB}[1]{\ensuremath{\mathbf{#1}}}
\newcommand{\send}{\MB{send}\xspace}

\newcommand{\errorbound}{\ensuremath{\|\hmin-\hmax\|_1+\agreedist}\xspace}
\newcommand{\postsmr}[1]{\ensuremath{\command{post_{\smr}}(#1)}\xspace}
\newcommand{\Value}{\ensuremath{\mathsf{VALUE}}\xspace}
\newcommand{\Vprop}{\ensuremath{\mathsf{VPROP}}\xspace}

\newcommand{\Vpost}{\ensuremath{\mathsf{VPOST}}\xspace}

\newcommand{\FTpost}{\ensuremath{\mathsf{FTPOST}}\xspace}

\newcommand{\Votevp}{\ensuremath{\mathsf{VOTEVP}\xspace}}

\newcommand{\Voteft}{\ensuremath{\mathsf{VOTEFT}\xspace}}
\newcommand{\quorum}{\ensuremath{\mathcal{QC}}\xspace}
\newcommand{\adiststartdate}{1-Oct-2022\xspace}
\newcommand{\adistenddate}{10-Oct-2022\xspace}
\newcommand{\simstartdate}{11-Oct-2022\xspace}
\newcommand{\simenddate}{10-Nov-2022\xspace}
\newcommand{\binance}{\textsf{Binance}\xspace}
\newcommand{\coinbase}{\textsf{Coinbase}\xspace}
\newcommand{\cryptocom}{\textsf{crypto.com}\xspace}
\newcommand{\ftx}{\textsf{FTX}\xspace}
\newcommand{\huobi}{\textsf{Huobi}\xspace}
\newcommand{\okcoin}{\textsf{OKCoin}\xspace}
\newcommand{\okex}{\textsf{OKEx}\xspace}

\newtheorem{property}{Property}

\theoremstyle{definition}
\newtheorem{definition}{Definition}

\def\cert{\ensuremath{\mathcal{QC}}\xspace}
\def\BibTeX{{\rm B\kern-.05em{\sc i\kern-.025em b}\kern-.08em
    T\kern-.1667em\lower.7ex\hbox{E}\kern-.125emX}}

\makeatletter
\newcounter{HALG@line}
\renewcommand{\theHALG@line}{\thealgorithm.\arabic{ALG@line}}
\makeatother

\begin{document}


\title{DORA: Distributed Oracle Agreement with Simple Majority\textsuperscript{**} 
\thanks{\textsuperscript{**} Author names are in alphabetical order of their last name.}}
\author{\IEEEauthorblockN{Prasanth Chakka\IEEEauthorrefmark{1},  Saurabh Joshi\IEEEauthorrefmark{2}\textsuperscript{\orcidicon{0000-0001-8070-1525}}, Aniket Kate\IEEEauthorrefmark{3}\textsuperscript{\orcidicon{0000-0003-2246-8416}}, Joshua Tobkin\IEEEauthorrefmark{4} and David Yang\IEEEauthorrefmark{5}\textsuperscript{\orcidicon{0009-0002-8101-9434}}}
\IEEEauthorblockA{\IEEEauthorrefmark{1}\IEEEauthorrefmark{2}\IEEEauthorrefmark{3}\IEEEauthorrefmark{4}\IEEEauthorrefmark{5}Supra Research \\
\IEEEauthorrefmark{3} Purdue University \\
\IEEEauthorrefmark{5} Hamad bin Khalifa University \\
}
}

\maketitle
\begin{abstract}
Oracle networks feeding off-chain information to a blockchain are required to solve a distributed agreement problem since these networks receive information from multiple sources and at different times. We make a key observation that in most cases, the value obtained by oracle network nodes from multiple information sources are in close proximity. We define a notion of agreement distance and leverage the availability of a state machine replication (SMR) service to solve this distributed agreement problem with an honest simple majority of nodes instead of the conventional requirement of an honest super majority of nodes.  
Values from multiple nodes being in close proximity, therefore, forming a coherent cluster, is one of the keys to its efficiency. Our asynchronous protocol also embeds a fallback mechanism if the coherent cluster formation fails.
 Through simulations using real-world exchange data from seven prominent exchanges, we show that even for very small agreement distance values, the protocol would be able to form coherent clusters and therefore, can safely tolerate up to $1/2$ fraction of Byzantine nodes. We also show that, for a small statistical error, it is possible to choose the size of the oracle network  
to be significantly smaller than the entire system tolerating up to a $1/3$ fraction of Byzantine failures. This allows the oracle network to operate much more efficiently and horizontally scale much better.

\end{abstract}

\nocite{*}
\section{Introduction}
Connecting existing Web 2.0 data sources to  blockchains is crucial for next-generation blockchain applications such as decentralized finance (DeFi). An oracle network~\cite{ACMSurvey,9086815} consisting of a network of independent nodes aims to address  this issue by allowing blockchain smart contracts to function over inputs obtained from existing Web 2.0 data, real-world sensors, and computation interfaces.

Performing this information exchange securely, however, is not trivial. First, only a few data sources may be available to pick from, and some of them may crash (due to a {\em Denial of Service} (DOS) attacks or system failures), or even send  incorrect information (due to a system compromise or some economic incentives).~\cite{Luna,CoinMarketCap} Second, as most of the data sources today do not offer data in a signed form, the oracle network  also becomes vulnerable due to the compromise of a subset of oracle network nodes, a subset of data sources, or due to the compromise of a combination of the two.  Third, the adversary may go after the availability of the system (and at times safety) by malevolently slowing down the protocols. We address these issues and propose a robust and scalable distributed solution for solving the oracle problem. 
While our approach can withstand extreme adversarial settings, we make some real-world synchrony and input-data distribution observations, and introduce oracle execution sharding. This makes our solution scale quite well as the number of oracle services and the size of the oracle network grows.

One of the key objectives of an oracle service is to take a piece of off-chain information and bring it to the on-chain world. Therefore, any such service must have three components, (i) sources of information, which we shall refer to as  data sources, (ii) network of nodes, and (iii) target component in an on-chain environment, for example, a smart contract. To maintain fault tolerance capabilities as well as the decentralized nature of the service, it is necessary to have multiple data sources in addition to having a network of multiple nodes. The oracle agreement problem focuses on producing a unique value that is representative of the values emanating from the honest data sources that are feeding information to the oracle network.  We need a protocol for ensuring that all the honest oracle nodes have the same output that is representative of all the honest data sources. Notice that this is non-trivial as all the honest nodes may still have different outputs since they may be listening to different sets of data sources at slightly different times. We call this problem a {\em Distributed Oracle Agreement} (\dora) problem. \dora shares the same termination and agreement properties with the well-studied {\em Byzantine agreement} (BA) problem~\cite{FLP85,LSP82}.  However, the crucial validity property for \dora is significantly more generic. BA demands that the output be the same as an honest node's input if all the honest nodes have the same input. \dora is a generalization of this, where the output will be a value within a range  defined by the minimum and maximum honest inputs. As \dora generalizes the BA problem, it requires the standard $67\%$ honest majority among the participating nodes. As the system scale and number of participating nodes increase, solutions to \dora may not scale, especially when we collect many different kinds of variables. For a full-fledged blockchain system, we expect the number of variables, for which representative values are to be computed, to be in the order of several hundred.

Let us first understand how one of the state-of-the-art oracle \sotaoracle protocol~\cite{ocr} works. An oracle network consisting of multiple nodes obtains different values from multiple sources of information as shown in \cref{Fi:oracle} and agrees on a single report/output. As we discuss in \cref{sec:DORADef}, this problem is closely related to the standard Byzantine agreement problem~\cite{jacm80lamport}. In presence of Byzantine nodes and non-synchronous communication links, the oracle network has to have at least $3f+1$ nodes, where $f$ is the upper-bound on the number of nodes that can turn Byzantine during the agreement protocol~\cite{jacm88}. These oracle nodes, run a Byzantine agreement protocol, to agree on a subset of $2f+1$ values. In \sotaoracle, one of the nodes is designated as a {\em leader} which would send this agreed-upon subset to the Blockchain. If the leader happens to be Byzantine, after a certain timeout period, the \sotaoracle protocol initiates a leader change.
\Omit{
\begin{figure}
	\centering
\begin{tikzpicture}
\tikzstyle{data source}=[rectangle,fill=black!25]
\tikzstyle{onodes}=[circle,fill=green!50]
\tikzstyle{oval}=[draw,ellipse]
\foreach \name / \y in {0,...,4}
 	\node[data source] (ds-\name) at (0,-\y) {};
	\node[above=0.5cm of ds-0] (dslabel) {Data Sources};

\foreach \name / \y in {0,...,4}
\node[onodes] (on-\name) at (3,-\y) {};

\foreach \dest in {0,...,4}
\draw[] foreach \x [evaluate=\x as \y using {mod(int(\x+\dest),5)}] in {0,...,2} {(ds-\y)--(on-\dest)};
\node[oval,minimum width=1cm, minimum height=5cm, below=-0.7cm of on-0] (comm)  {};
\node[right=-0.1cm of comm,rotate=90,yshift=0.2cm] (commlabel) {committee};
\node[oval,minimum width=5cm, minimum height=6cm, right=-3cm of comm,xshift=1cm] (tribe) {};
\node[above=-0.5cm of tribe] (tribelabel) {tribe};
\node[right=3.5cm of tribe.north] (tail) {};
\node[right=3.5cm of tribe.south] (head) {};
\draw[-{Triangle[width=18pt,length=8pt]}, line width=10pt](tail)--(head) node[midway] (smrmid)  {};
\node[left=0.1cm of smrmid,rotate=90] (smrlabel) {SMR Chain};
\draw[->] let \p1=(on-0), \p2=(smrmid) in (on-0)--(\x2-5,\y1);
\draw[->] let \p1=(on-2), \p2=(smrmid) in (on-2)--(\x2-5,\y1);

\end{tikzpicture}
\caption{Oracle Model}
\label{Fi:oracle}
\end{figure}
}

\begin{figure}
	\centering
 	\includegraphics[scale=0.4]{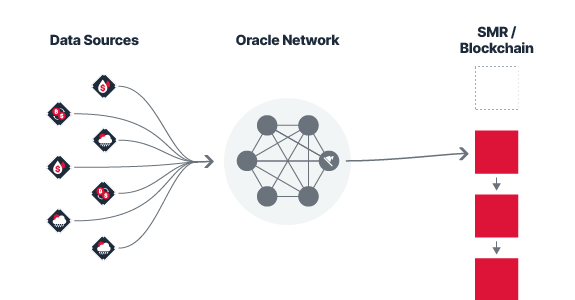}

	\caption{\sotaoracle protocol: Network of oracle nodes obtaining data from data sources, running a Byzantine agreement protocol and then broadcasting the agreed upon value to the world via SMR/Blockchain  }
	\label{Fi:oracle}
\end{figure}

A crucial point to note is that in \sotaoracle, the SMR/Blockchain is used merely as a means to broadcast/publish the agreed-upon set to the world.  Even without the SMR/Blockchain, \sotaoracle protocol ensures that all the honest nodes of the oracle network agree upon the exact same subset of $2f+1$ values. In the network with $f$ Byzantine nodes, $2f+1$ values are required to ensure that statistical aggregation via computing the median does not deviate {\em too much}. Essentially, $2f+1$ values ensure that the median is upper-bounded and lower-bounded by values from the honest nodes.

Blockchain ensures that the messages on the chain have a total order. Further, it also ensures that any entity observing the state of a Blockchain would witness the exact same total order. Not utilizing the ability of Blockchain to act as an ordering service is a missed opportunity by the oracle networks in general.

\subsection{Our Approach}
We leverage the ordering capabilities of the Blockchain by presenting a protocol that both sends and receives information from the Blockchain to accomplish the goal of publishing a value on the Blockchain that is representative of the values from all the data sources. Towards computing a representative value, we redefine the notion of agreement. We say that two nodes {\em agree} with each other if the values that they obtained from data sources are within a pre-defined parameter called {\em agreement distance}. We say that a set of values form a {\em coherent cluster} if all the values in that set are at most agreement distance away from one another. The oracle network now merely needs to agree on a coherent cluster of size $f+1$. Since there would be at least one honest value in any such coherent cluster, any statistical aggregator such as the mean or the median would be at most agreement distance away from an honest value, thus ensuring that the final agreed upon value does not deviate {\em too much} from an honest value.

The way our protocol achieves the agreement on a coherent cluster of size $f+1$ is that the first such cluster posted on the Blockchain would be considered authoritative for a given round of agreement.

Just like \sotaoracle protocol, we also have a designated node that we call an {\em aggregator}. These aggregators would gather $f+1$ signatures on a proposed set of $f+1$ values provided by the nodes and then post it on the SMR/Blockchain. To circumvent the problem of having a Byzantine aggregator, we sample a set of aggregators, henceforth called a family of aggregators, from the entire oracle network such that there is at least one honest aggregator. This helps us avoid having to initiate any aggregator change protocol. So now, we have multiple aggregators posting sets of size $f+1$ to the Blockchain, but only the first one would be accepted as authoritative by the oracle network. Note that all such coherent clusters of size $f+1$ contain values that are at most agreement distance away from some honest value.

With the redefined notion of agreement, and utilizing ordering capabilities of the Blockchain, we reduce the number of oracle nodes required to $2f+1$  instead of the usual requirement of $3f+1$.
\begin{figure}
	\centering
	\includegraphics[scale=0.5]{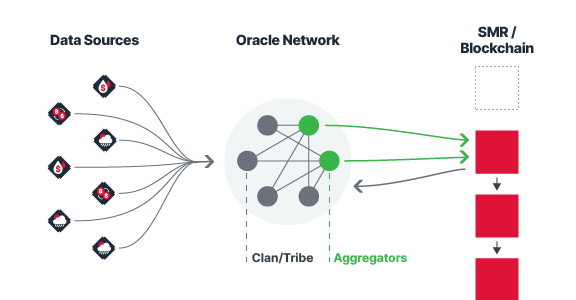}

	\caption{The oracle nodes collect information from various data sources. The oracle nodes exchange information with the aggregators to vote for the proposal of a coherent cluster of size $f+1$. The aggregators post these clusters to Blockchain and all oracle nodes consider the first cluster to be authoritative.}
	\label{Fi:supraoracle}
\end{figure}

\cref{Fi:supraoracle} shows how information flows across various components of our oracle protocol. Note that now the oracle nodes only have to communicate back and forth with the family of aggregators. All the oracle nodes and the aggregators observe the total order in which information appears on the SMR/Blockchain. Aggregators are the only set of nodes that sends information to the Blockchain.

Under unusual circumstances when none of the aggregators is able to form a coherent cluster, our protocol switches to a fallback mechanism. In this fallback mechanism, the requirement of total nodes again increases to $3f+1$ out of which the aggregators wait for $2f+1$ nodes to provide a value and compute a median. In this fallback mechanism, the arithmetic mean can not be used anymore since the values introduced by the Byzantine nodes in this set of size $2f+1$ could be unbounded. Therefore, our fallback mechanism, along the same lines as \sotaoracle, computes a median from the agreed upon set of size $2f+1$.

Our idea of the reduced requirement for the number of nodes, hinges upon the honest nodes producing values in close proximity most of the time. To prove that this assumption is practical and well-founded, we take \BTC price information from $7$ exchanges, that include \ftx, and run a simulation of our protocol on this data covering a $30$ day period that includes the turbulent period of \ftx collapse~\cite{ftxcollapse}. We run simulations where the round of agreement happens every $30$ and $60$ seconds. These simulations corroborate our close proximity of honest values assumption. We observe that for agreement distance as low as $\$25$ and $\$53$, when the average \BTC price is around $\$19605$, a coherent cluster can be formed in $93\%$ and $99\%$ rounds of agreement. 

To further scale our solution, we introduce execution sharding by sampling multiple sub-committees from the $3f+1$ nodes available in the entire system. The size and the number of such sub-committees are chosen such that each sub-committee has an honest simple majority with a very high probability. Since under normal circumstances we only need an honest simple majority, we equally divide the responsibilities to track multiple variables and bring their price information to the chain equally among these sub-committees. Since a sub-committee can now only provide a probabilistic safety guarantee, we analyze it in \cref{Se:probsafety} and establish that with a few hundred nodes, the safety guarantee holds with a very high probability.

The contributions of this paper are as follows:

\begin{compactitem}
\item We redefine the agreement notion by introducing a parameter called {\em agreement distance} (\cref{Se:prelim}).
\item We propose a novel protocol that achieves agreement with $2f+1$ nodes, where $f$ of the nodes could be Byzantine (\cref{Se:protocol}).
\item We propose a multi-aggregator model in our protocol to ensure liveness without introducing extra latency of aggregator/leader change (\cref{Se:protocol}).
\item We demonstrate via empirical analysis of real-world data (\cref{Se:empirical}) that most of the time the honest nodes would be able to provide values in very close proximity, therefore, the protocol can function with $2f+1$ nodes with very small margins of potential error.
\item Probability analysis of effects of the size of the network of nodes and the size of the family of aggregators on the safety of the protocol is provided in  \cref{Se:probsafety}. It is evident that with a few hundred nodes in the entire system, and with $10-15$ nodes in the family of aggregators, the protocol can function safely with very high probability.
\end{compactitem}

\section{Preliminaries \label{Se:prelim}}
In this section, we introduce some preliminaries and notations to which we will refer for  the remainder of this paper.

\subsection{Oracle Network}

Let $|S|$ denote the cardinality of a set $S$.
We shall use $\median{X}$ to denote the median of a set $X$ of values.
Our oracle network consists of a set $N_t$ of $|N_t|$ nodes, which we also call a tribe.
Among these, at most $f_t< \frac{|N_t|}{3}$ nodes may get compromised, that is to turn Byzantine, thus,  deviating from an agreed-upon protocol and behaving in an arbitrary fashion. A node is \textit{honest} if it is never Byzantine.
We assume a static adversary that
corrupts its nodes before the protocol begins. 

From these $N_t$ nodes, towards developing scalable solution, we also uniformly randomly draw a sub-committees $N_c$, henceforth called a clan,  of size $|N_c|$ such that at most $f_c< \frac{|N_c|}{2}$ nodes are Byzantine in a clan. The clan nodes are sampled uniformly at random from the tribe, and we ensure that the statistical error probability of a clan having more than $\lfloor \frac{|N_c|}{2}\rfloor$ Byzantine nodes is negligible. Probability analysis for such sampling is provided in \cref{Se:probsafety}.

All oracle nodes are connected by pair-wise authenticated point-to-point links. We assume this communication infrastructure to be asynchronous such that the (network) adversary can arbitrarily delay and reorder messages between two honest parties. As typical for all asynchronous systems, for the system's liveness properties, we assume that the adversary cannot drop messages between two honest parties.

 We initialize the protocol with a unique identifier to prevent replay attacks across concurrent protocol instances but do not explicitly mention this in the text.

A signed messages $m$ from a node $p_{i}$ are denoted by $m{(\cdot)}_{{i}}$. Similar to most recent SMR/blockchain designs, we assume a $(n,n-f)$ threshold BLS~\cite{boneh2001short,BachoL22} signature setup. We denote an $n-f$ threshold signature on the message $m$ as a quorum certificate $\cert{m}$.

While the use of threshold signatures offers a simple abstraction 
and can be verified on Ethereum,
we can also employ a multi-signature (multisig) setup allowed by 
the employed blockchain.

\subsection{Oracle Data Sources}

A data source is responsible for providing the correct value of a variable \commodity, which could be, say, the price of Bitcoin in US Dollars.
Let $DS$ denote a set of data sources and $BDS \subset DS$ denote the subset of these data sources which could be Byzantine. We say that a data source is Byzantine if : (i) it lies about the value of the variable, or (ii) if it is non-responsive. Otherwise, we will consider
the data source to be honest. 
We assume that $|BDS|\leq f_d$, where $f_d$ is the upper bound on the number of Byzantine data sources.

The goal of the tribe is to reach a consensus in a distributed fashion about a {\em representative value}, denoted as \svalue, of a particular \commodity. The notion of a representative value depends 
on the underlying \commodity. For example, we can say that the representative value 
of a particular stock could be considered to be a mean of the values of the stock at various stock exchanges.

In this paper, for now, we assume that for the variable $\commodity$ of interest, the arithmetic mean $\mu$  of values of \commodity from various data sources  is the representative value. 

An observation of a node $p_i$ from a data-source $ds_j$ of a variable $\commodity$  is
denoted as $o(p_i,ds_j,\commodity)$. We say that an observation $o(p_i,ds_j,\commodity)$ is an
honest observation if both $p_i$ and $ds_j$  are honest.  In an ideal world without any Byzantine nodes or Byzantine data sources, we would like the protocol followed by the oracle network to have the following property:

\begin{property}[Ideal Representative Value] \label{prop:wmean} $\svalue = \frac{\sum_{ds_j \in DS}  o(p_i,ds_j,\commodity)}{|DS|}$ \\
	where $p_i$ refers to any one of the honest nodes.
\end{property}

When the context is clear, we will just use $o$ to denote an observation. 
Let $O$ denote the set of all observations. We will use $O_{p_i}$ to denote the observations made by node ${p_i}$.
$\honest(O)$ and $\Byz(O)$ denote the set of honest observations and Byzantine observations respectively.
Let $\hmin(O_{p_i}) = \min_{\honest(O_{p_i})}$ and $\honest_{max}(O_{p_i}) = \max_{\honest(O_{p_i})}$ indicate the minimum and maximum values
from a given set of honest observations $\honest(O_{p_i})$. We will just use $\hmin$ and $\honest_{max}$ to refer to the minimum and the maximum values
amongst all honest observations $\honest(O)$.

We say that two observations $o_1$ and $o_2$ agree with each other if $\agree{o_1}{o_2}$. That is if the $L_1$ distance between them is at most $\agreedist$, where $\agreedist$ is a pre-defined parameter known as {\em agreement distance}.
A set of observations $CC \subseteq O$ is said to form a {\em coherent cluster} if $\forall_{o_1,o_2\in CC}: \agree{o_1}{o_2} $.

We will use the terms {\em majority} and {\em super majority} to denote that some entity has a quantity strictly greater than $\frac{1}{2}$  and $\frac{2}{3}$ fraction of the total population respectively.
For example, an honest majority within a set of nodes would indicate that more than half of the nodes are honest. An honest super-majority would similarly indicate the fraction of honest nodes to be strictly greater than $\frac{2}{3}$ of all the nodes within the set.

Let $\svalue_r$ denote the value for which the oracle network achieved a consensus for it to be considered as the representative value for a round $r$. We will use $\svalue_{r-1}$ to denote the value emitted by the oracle network in the previous round.

\subsubsection{Expected Representative Value with Byzantine Actors}
We consider two kinds of bad actors in the system: (i) Byzantine data sources, and (ii) Byzantine nodes.
We assume that Byzantine nodes and Byzantine data sources can collude in order to 
 (i) prevent the oracle network from reaching a consensus in a given
 round, or (ii) prevent the oracle network from achieving \cref{prop:wmean}.

In the presence of such Byzantine actors, even if a single dishonest observation gets considered for computation of the mean to compute \svalue, the Byzantine nodes can be successful in setting \svalue to deviate
arbitrarily from the true representative value equivalent to a mean\footnote{Byzantine oracle nodes are problematic particularly as data sources do not sign their inputs.}. 

Moreover, for the setting with $|N_t| \geq 3f_t+1$ nodes, a rushing adversary can suggest its input only after observing the honest parties' inputs. Therefore, we can only aim for the following weaker property: 

\begin{property}[Honest Bounded Value] \label{prop:hbounds} $H_{min} \leq \svalue \leq H_{max}$
 \end{property} 
 
 As discussed in the literature~\cite{mendes2013approx}, the agreed value is in the convex hull of the non-faulty nodes' inputs.

\subsection{Primitive: State Machine Replication}\label{sec:SMR}

In the form of state machine replication (SMR)~\cite{Schneider90} or blockchain, we employ a key distributed service for our system. An SMR service employs a set of replicas/nodes collectively running a deterministic service that implements an abstraction of a single, honest server, even when a subset of the servers turns Byzantine. In particular, an SMR protocol orders messages/transactions $tx$ from clients (in our case aggregators) into a totally ordered log that continues to grow.
We expect the SMR service to provide public verifiability.
Namely, there is a predefined Boolean function \command{Verify}; a replica or a
client outputs a log of transactions log $= [ tx_0, tx_1, \dots,tx_j ]$ if and only
if there is a publicly verifiable proof $\pi$ such that $ \command{Verify}(log, \pi) = 1$.

Formally, an SMR protocol~\cite{momose2021multi} then provides the following safety and liveness:
\begin{property}[Safety] \label{prop:smrsafety}
	If $[ tx_0, tx_1, \dots,tx_j]$ and $[ tx'_0, tx'_1, \dots,tx'_{j'}]$ are output by
two honest replicas or clients, then
$tx_i = tx'_i$
for all $i \leq \min(j,j')$.

\end{property}
\begin{property}[Liveness] \label{prop:smrliveness}
If a transaction $tx$ is input to at least one honest
replica, then every honest replica eventually outputs a log containing $tx$.

\end{property}

We assume that there is an SMR/blockchain service running in the background that an oracle service can employ. 
Oracle network nodes are assumed to employ a simple put/get interface to the SMR service. They employ $\postsmr{\cdot}$ to post some (threshold) signed information (or transaction) on the SMR chain. Upon collecting and processing ordered transactions on SMR, the nodes employ ``{\bf Upon} witnessing''  event handling to process the relevant messages. As communication links between oracle nodes and SMR service nodes are expected to be asynchronous, this put/get interface is expected to function completely asynchronously and provide guarantees that can be observed as an interpretation of SMR safety and liveness: (i) senders' messages appear on the blockchain eventually; (ii) different receivers observe messages at different points in time; (iii) however, all the nodes eventually observe messages in the exact same total order.

Notice that some recent works view blockchains as broadcast channels. In such cases, a sender node's message is expected to be delivered to all correct receiver nodes within a predefined number of blocks. If the receivers fail to deliver any message by the end of this period, the sender is treated as faulty. However, determining an appropriate time bound for any blockchain is challenging due to occasional network asynchrony over the Internet, as well as transaction reordering, frontrunning~\cite{274723}, and eclipse attacks~\cite{MarcusHG18}. Even if a pessimistic bound can be established, it may be significantly worse than the usual computation and communication time. For example, in practice, the lightning network based on Bitcoin uses a time bound of approximately one day.~\cite{AumayrMKM23} While newer blockchains with shorter block periods may be able to reduce the time-bound, the asynchronous primitive described above remains a better representation of modern blockchains.

\section{The Oracle Protocol}

\subsection{Data Feed Collection}
The first step of an oracle service involves  oracle nodes connecting with the data sources. As we assume that there are multiple data sources, some of them can be compromised and  most of them do not sign their data feeds. Our key goal is to ensure that the honest (oracle) node's input to aggregators is representative of the honest data sources.

Towards ensuring the correctness of the honest nodes' inputs, we expect it to retrieve feeds from multiple data sources such that the median of the received values is representative of the honest values; i.e., it is inside the $[\hmin,\hmax]$ range of the honest data sources. 

\ifasyncdstonode
The data feed collection works as shown in \cref{Alg:datafeed}

\begin{enumerate}
\item We expect that up to $f_d$ data sources may become Byzantine.
 
		Therefore, out of abundant caution, we mandate that every oracle node connects to $3f_d+1$ data sources.
	\item Every node sends a request to their assigned set of data sources $ADS$. (\crefrange{Li:dfinitstart}{Li:dfinitend})
	\item Whenever a value $v$ is received from a data source $ds$, the node stores in $obs[ds]$ . (\cref{Li:dfgetvalue}) If at least $2f_d+1$ values have been received, gather all the values received in $Obs$. {\cref{Li:dfreportmedian}}
	
\end{enumerate}
\begin{theorem}
[Termination of Data Feed][restate]\label{The:dfterminate} \cref{Alg:datafeed} terminates eventually.
\end{theorem}
\begin{proof}
    Since every node is listening to $3f_d+1$ data sources and at most $f_d$ data sources could be Byzantine, each node is guaranteed to receive values from $2f_d+1$ data sources (\cref{Li:dfsupermajority}) thus resulting in termination of data feed phase for every node.
\end{proof}
\else
In this direction, we make a key synchrony assumption about communication links between data sources and oracle nodes. Unlike communication links between oracles nodes, we assume that links connecting data sources and oracle nodes to be bounded-synchronous such that if a node does not hear back from the data source over the web API/socket in a time-bound $\dstimebound$, the node can assume that the source has become faulty/Byzantine.

In this bounded-synchronous communication setting, data feed collection works as shown in \cref{Alg:datafeed}.

\begin{enumerate}
	\item We expect that up to $f_d$ data sources may become Byzantine.
 
		Therefore, out of abundant caution, we mandate that every oracle node connects to $2f_d+1$ data sources \footnote{Our bounded-synchrony assumption for the source-to-node link is for simplicity. If these links behave more asynchronously, we can easily make the node contact $3f_d+1$ data sources and wait to hear back from at least $2f_d+1$ data sources to ensure that the honest nodes select a representative value.}.
	\item Every node sends a request to their assigned set of data sources $ADS$ and then starts the timer $\dstimebound$. (\crefrange{Li:dfinitstart}{Li:dfinitend})
	\item Whenever a value $v$ is received from a data source $ds$, the node stores in $obs[ds]$  . (\cref{Li:dfgetvalue})
	\item Upon {\em timeout} of $\dstimebound$ gather all the values received so far in $Obs$. (\cref{Li:dfreportmedian})
\end{enumerate}
\fi
\begin{theorem} \label{The:median} At the end of \Cref{Alg:datafeed}, $\hmin(Obs) \leq \median{Obs} \leq \hmax(Obs)$.
\end{theorem}
\begin{proof}
	We will prove the theorem by contradiction. Without the loss of generality, let us assume that $\median{Obs} > \hmax(Obs)$. The total number of all the honest values received by a node is at least $f_d+1$ and the total number of all the Byzantine data sources is at most $f_d$ ($\Byz(Obs) \leq f_d$). Since all the honest data sources would report the correct value
  \ifasyncdstonode
 Since all the honest data sources would report the correct value,
 \else
 Since all the honest data sources would report the correct value
 within $\dstimebound$,
 \fi
 we would have $f_d+1 \leq |Obs|$. Since $\median{Obs} > \hmax(Obs)$ and $|\honest(Obs)| \geq f_d+1$, it must be the case that $|\Byz(Obs)| \geq  (f_d+1)$, a contradiction. One can similarly argue that $\median{Obs} < \hmin(Obs)$ is not possible.
\end{proof}

\ifasyncdstonode
\begin{algorithm}
	\caption{GetDataFeed($ADS$)}
\label{Alg:datafeed}
\algblockdefx[UPONBLOCK]{UPONSTART}{UPONEND}
[1]{\textbf{Upon} #1 \textbf{do} }{}

\begin{algorithmic}[1]
	\Require $\left(ADS \subseteq DS\right) \wedge \left(|ADS| = 3f_d+1\right)$ \;
	\UPONSTART{init} \label{Li:dfinitstart} \;
\State  $\forall_{ds \in ADS} obs[ds] \gets \bot$ \;
\State send the request to all $ds \in ADS$ \label{Li:dfinitend} \;
\UPONEND 

\UPONSTART{receiving value $v$ from $ds$ }
\State {$obs[ds] \gets v$} \label{Li:dfgetvalue}
\If{$|\{x | x \in obs \wedge x \neq \bot\}| \geq 2f_d+1$} \label{Li:dfsupermajority}
\State \Return \,$Obs=\{obs[ds] | obs[ds] \neq \bot\}$ \label{Li:dfreportmedian}\;
\EndIf
\UPONEND

\end{algorithmic}
\end{algorithm}
\else
\begin{algorithm}
	\caption{GetDataFeed($ADS$)}
\label{Alg:datafeed}
\algblockdefx[UPONBLOCK]{UPONSTART}{UPONEND}
[1]{\textbf{Upon} #1 \textbf{do} }{}

\begin{algorithmic}[1]
	\Require $\left(ADS \subseteq DS\right) \wedge \left(|ADS| = 2f_d+1\right)$ \;
	\UPONSTART{init} \label{Li:dfinitstart} \;
\State  $\forall_{ds \in ADS} obs[ds] \gets \bot$ \;
\State send the request to all $ds \in ADS$ \;
\State start the timer $\dstimebound$ \label{Li:dfinitend} \;
\UPONEND 

\UPONSTART{receiving value $v$ from $ds$ \textbf{and} $\dstimebound > 0$}
\State {$obs[ds] \gets v$} \label{Li:dfgetvalue}
\UPONEND

\UPONSTART{${\mathit timeout}$ of $T_{ds}$}
\State \Return $Obs=\{obs[ds] | obs[ds] \neq \bot\}$ \label{Li:dfreportmedian}
\UPONEND
\end{algorithmic}
\end{algorithm}
\fi

\subsection{Distributed Oracle Problem Definition}\label{sec:DORADef}
Once we ensure that every honest oracle node has produced a correct/representative value as its input, the oracle problem becomes a bit simpler.
Since honest observations may still differ, therefore, we need to make sure that the honest nodes agree on exactly the same value, which is again  representative of the honest nodes' inputs. We observe that this problem is related to the Byzantine agreement (BA)~\cite{jacm88} and Approximate agreement~\cite{jacm86dolev,mendes2013approx} problems from the literature on distributed systems. While the expected agreement and termination properties are exactly the same as for BA, the validity property coincides with the typical Approximate agreement definition. We call this problem a {\em distributed oracle agreement} (DORA) problem.

\begin{definition}[Distributed Oracle Agreement---DORA]\label{def:dora}
    A distributed oracle agreement (DORA) protocol among $n$ nodes $\{p_1,p_2,\dots,p_n\}$ with each node having input $v_i$ guarantees the following properties:
    \begin{property}[Termination] \label{prop:doratermination}
 All honest nodes eventually decide on some value.

    \end{property}
    \begin{property} [Agreement] \label{prop:doraagreement}
The output value \svalue for all honest nodes is the same.

    \end{property}
    \begin{property}[(Min-max) Validity] \label{prop:doravalidity}
	    The output value is in the convex hull of the honest nodes' inputs. For scalar inputs, 
        this coincides with  \cref{prop:hbounds}: $H_{min} \leq \svalue \leq H_{max}$

    \end{property}
    \end{definition}

Similar to the BA problem, it is easy to observe that DORA also requires an honest super majority. It is however interesting to observe that this bound persists even when the oracle nodes have access to the SMR/blockchain service defined in \cref{sec:SMR}. Unlike a typical broadcast channel, this  SMR channel is asynchronous to different receiving nodes. Therefore, it is not possible to differentiate between slow nodes and crashed nodes. The protocol needs to make progress with only $n-f$ inputs, where $f$ out of $n$ nodes are Byzantine. Access to the SMR service is still helpful as it already overcomes the FLP impossibility~\cite{FLP85}. We can develop protocols for BA and DORA in a purely asynchronous manner without requiring any distributed randomness (such as common coins)~\cite{CommonCoin}. 

The requirement of an honest super-majority for DORA can be a scalability issue as the number of oracle nodes increases. Based on our analysis of real-world data, for optimistic scenarios, we overcome this issue by making a practical assumption on the input values. In particular, if we assume that inputs from all honest nodes form a coherent cluster within a reasonably small agreement distance $\agreedist$, then we can solve \dora requiring only an honest majority instead of an honest super majority. 
We call this new problem  \doracc which may offer a slightly weaker validity property stated below:
\begin{definition}[Distributed Oracle Agreement with Coherent Cluster (\doracc)]
	A \doracc  protocol among $n$ nodes $\{p_1,p_2,\dots,p_n\}$ with each node having input $v_i$ such that these input values form a coherent cluster for a distance $\agreedist$ guarantees the following property in addition to Termination(\cref{prop:doratermination}) and Agreement(\cref{prop:doraagreement}) properties.
	\begin{property}[Approximate (Min-max) Validity.] \label{prop:doraccvalidity}  The output \svalue is within the interval $[\hmin-\agreedist,\hmax+\agreedist]$.
	\end{property}
\end{definition}

\subsection{DORA with Coherent Cluster Protocol \label{Se:protocol} }
\doracc  only needs an honest majority, while we expect the tribe to offer honest nodes with a supermajority.
We assume all the nodes to know \agreedist a priori, which is a parameter to the protocol.
In an optimistic scenario, this allows us to only employ a subset of nodes within the tribe.
We divide the tribe  $N_t$ into multiple mutually exclusive clans of size $|N_c|$, where $N_c$ denotes the set of nodes belonging to a clan $c$. While the number of Byzantine nodes within $N_t$ is restricted to $f_t$, we choose $|N_c|$ such that the number of Byzantine nodes are at most $f_c = \frac{|N_c|-1}{2}$ with very high probability. Each such clan of size $|N_c|$ can be given the responsibility to emit \svalue-values for  different variables. For simplicity, we only focus on a single variable \commodity in this paper. The process, however, can be replicated for multiple variables. Note that \agreedist could be different for different variables.

There could be times when the inputs from all honest nodes may not form a coherent cluster within the distance $\agreedist$. For such a scenario, we aim at first identifying this volatility in a distributed fashion and then securely switching to the fallback protocol for the \dora instance over the entire tribe. Now, when we run the \dora instance over the entire tribe, we ensure that the output satisfies \cref{def:dora} and is representative of the existing market conditions.

Additionally, we also uniformly randomly select a {\em family} $\aggset$ of nodes from the tribe  such that at least one of them is honest with high probability. These  nodes are designated as aggregators as they are supposed to securely collect information from the clan nodes and post it on the \smr. Note that these aggregators are employed solely to reduce the total number of interactions with the blockchain. Since the nodes sign their inputs, there are no safety issues regarding a Byzantine aggregator forging them. If there is only a single aggregator, there is almost $\frac{1}{3}$ probability of it being Byzantine, which may require an aggregator change to ensure the progress of the protocol. Our multi-aggregator model ensures progress without requiring any aggregator changes, thus reducing latency. However, this comes with added communication complexity.

\begin{algorithm}
	\caption{$Compute$\svalue$(p_i,N_c,\aggset,r)$}
\label{Alg:defaultfallback}

\algblockdefx[UPONBLOCK]{UPONSTART}{UPONEND}
[1]{\textbf{Upon} #1 \textbf{do} }{}

\begin{small}
\begin{algorithmic}[1]

	\State	\textbf{input:} $r$ is the round identifier, $N_c$ is the set if nodes in the clan, $\aggset$ is the set of aggregators, $p_i$ is the ID of this node
\UPONSTART{init}
\State $ADS \gets$ $2f_d+1$ uniformly randomly assigned data-sources from $DS$
\State $O_{p_i} \gets GetDataFeed(ADS)$
\State \send  $\Value(\median{O_{p_i}},r)_i$ to all nodes in \aggset
\State \hldiff{start timer $\fallbacktimer$} \label{Li:dpvinitftimer}
\If {$p_i \in \aggset$}
$\forall_{p_j \in N_c} obs[p_j] \gets \bot$
\EndIf
\UPONEND

\UPONSTART{receiving $\Value(v,r)_j$ from a node $p_j$ \textbf{and} $p_i \in \aggset$}
\State $obs[p_j] \gets v$
\If {$\exists_{CC \subseteq obs} |CC|\geq f_c+1$}
\State $\mu \gets mean(CC)$ \label{Li:dpvmu}
\State \send $\Vprop(CC,\mu,r)_i$ to all nodes in $N_c$ \label{Li:dpvsendvprop}
\EndIf
\UPONEND

\UPONSTART{receiving $\Vprop(CC,\mu,r)_j$ from $p_j \in \aggset$ } \label{Li:dpvsignedclusterstart}
\If {$Validate(\Vprop(CC,\mu,r)_j)=\true$}
 \send  $\Votevp(CC,\mu,r)_i$ to $p_j$ \label{Li:dpvsignedclusterend}
\EndIf
\UPONEND

\UPONSTART{receiving  $\Votevp(CC,\mu,r)_j$ from $p_j$ \textbf{and} $p_i \in \aggset$} 
\If {\quorum is formed  on $\Votevp(CC,\mu,r)$ } \Comment{Quorum with $f_c+1$ votes}
\State $\postsmr{\Vpost(CC,\mu,r,\quorum)}$ \label{Li:dpvpostmusmr}
\EndIf
\UPONEND
\UPONSTART{witnessing the {\bf first} $\Vpost(CC,\mu,r,\quorum)$ on \smr} \label{Li:dpvwitnessmu}
\State $\svalue_r \gets \mu$
\State \Return
\UPONEND

\UPONSTART{\hldiff{timeout on $\fallbacktimer$}}
\State \hldiff{\send   $\Voteft(fallback,r)_i$ to all nodes in $\aggset$} \label{Li:dpvnodefallback}
\UPONEND

\UPONSTART{\hldiff{receiving $\Voteft(fallback,r)_j$ from $p_j \in N_c$ \textbf{and} $p_i \in \aggset$}} \label{Li:dpvreceivefallback}
\If{\hldiff{$\quorum$ is formed on $\Voteft(fallback,r)$ \textbf{and} $\fallbacktimer$ for $p_i$ has already timed out }}
\State \hldiff{$\postsmr{\FTpost(fallback,r,\quorum)}$} \label{Li:dpvtimeoutfallbacksmr}
\EndIf
\UPONEND

\UPONSTART{\hldiff{witnessing $\FTpost(fallback,r,\quorum)$ on \smr}} \label{Li:dpvwitnessfb}
\State \hldiff{switch to Fallback protocol}
\UPONEND
\end{algorithmic}
\end{small}
\end{algorithm}
Pseudocode for \dora-CC in the presence of volatile data feeds is given in \cref{Alg:defaultfallback}.
It proceeds as follows:
\begin{enumerate}
	\item Using \cref{Alg:datafeed}, every node would gather data from $2f_d+1$ uniformly randomly assigned data sources, compute the median from all the values received and send the median to all the aggregators as a $\Value$ message. Each node starts a timer with $\fallbacktimer$ (\cref{Li:dpvinitftimer}) .
	\item An aggregator waits until a coherent cluster of size $f_c+1$ is formed. Once the cluster is formed, it computes the mean (\cref{Li:dpvmu}). The aggregator then sends the set of $\Value$ messages that formed a cluster along with the mean as a $\Vprop$ message to all the clan nodes (\cref{Li:dpvsendvprop}). This message would convey that the aggregator proposes the $\mu$ to be the $\svalue_r$.
	\item Upon receiving a proposal $\Vprop$ with $CC$ and $\mu$ from an aggregator, the node performs a validation. The node would validate that (i) $CC$ contains signed messages from the nodes of the clan, (ii) the values in $CC$ forms a coherent cluster within $\agreedist$, and (iii) that $\mu$ is indeed the arithmetic mean of the values in $CC$. If the validation is successful, it sends its signed vote $\Votevp(CC,\mu,r)$ to the aggregator (\crefrange{Li:dpvsignedclusterstart}{Li:dpvsignedclusterend}).

	\item In an optimistic scenario, the aggregator would receive $f_c+1$ votes of the kind $\Votevp(CC,\mu,r)$ approving the proposal to allow it to form a quorum,  prepare quorum certificate $\quorum$ and post $\Vpost(CC,\mu,r,\quorum)$ on the \smr (\cref{Li:dpvpostmusmr}). The nodes in turn would witness $\Vpost(CC,\mu,r,\quorum)$ on the \smr, agree on $\mu$ as $\svalue_r$ and conclude the current round (\cref{Li:dpvwitnessmu}).

	\item In an unusual scenario, a coherent cluster of size $f_c+1$ can not be formed by any of the aggregators. This could happen either due to extreme volatility during data feed, or due to network asynchrony and high network delays.
		In such a case,  the nodes would timeout on $\fallbacktimer$ and vote for fallback $\Voteft$ to all the aggregators  (\cref{Li:dpvnodefallback}). This can happen due to  high volatility during the data feed window $\dstimebound$ when the honest observations can not form a coherent cluster of size $f_c+1$.
		 In that case, the aggregators may receive fallback votes $\Voteft$ from multiple nodes (\cref{Li:dpvreceivefallback}). If an aggregator  has gathered $f_c+1$ votes for the fallback, it posts the fallback proposal with a quorum certificate on \smr (\cref{Li:dpvtimeoutfallbacksmr}). In a theoretical asynchronous model with the adversary having the ability to delay messages arbitrarily, the protocol would always fallback. However, we make a practical assumption that such selective delay of all honest nodes' messages may be uncommon, therefore, allowing us to exploit the proximity of the values that honest nodes emit.
	\item A node witnessing a fallback message  $\FTpost$ (due to timeout of $\fallbacktimer$) on \smr would switch to the fallback protocol mentioned in \cref{Alg:fallbackprotocol}.
\end{enumerate}

\begin{theorem}[Approximate Validity] \label{The:deltadistance} The protocol defined in \cref{Alg:defaultfallback} if it posts  $\Vpost(CC,\mu,r,\quorum)$ on \smr then $\mu$  must be within the interval $(\hmin-\agreedist,\hmax+\agreedist)$.
\end{theorem}
\begin{proof}
	There is at least one honest observation within the $f_c+1$ observations belonging to the cluster $CC$.
	By definition of a coherent cluster, all other $f_c$ observations, even if reported by Byzantine nodes, have to agree with the honest observation mentioned earlier to form the cluster $CC$. Therefore, no observation can exceed $\hmax+\agreedist$. Similarly, one can argue that no observation can be less than $\hmin-\agreedist$. Therefore, the mean $\mu$ computed by an aggregator must lie within the interval $(\hmin-\agreedist,\hmax+\agreedist)$.
\end{proof}
\begin{algorithm}
	\caption{$Fallback$\svalue$(p_i,r,N_t,\aggset)$}
	\label{Alg:fallbackprotocol}
\algblockdefx[UPONBLOCK]{UPONSTART}{UPONEND}
[1]{\textbf{Upon} #1 \textbf{do} }{}

\begin{algorithmic}[1]

	\State	\textbf{input:} $p_i$ is the unique id of this node, $r$ is the round identifier, $N_t$ is the set of nodes in the tribe, $\aggset$ is the set of aggregators
	\UPONSTART{init} \label{Li:fpinitstart}
\State $ADS \gets$ $2f_d+1$ uniformly randomly assigned data-sources from $DS$
\State $O_{p_i} \gets GetDataFeed(ADS)$
\State \send $\Value(\median{O_{p_i}},r)_i$ to all the aggregators (nodes in \aggset)
\If {$p_i \in \aggset$}
\State $\forall_{p_j \in N_c} obs[p_j] \gets \bot$
\EndIf \label{Li:fpinitend}
\UPONEND

\UPONSTART{receiving $\Value(v,r)_j$ from  $p_j$ \textbf{and} $p_i \in \aggset$ }
\State $obs[p_j] \gets v$
\State $O \gets  \{obs[p_j] | p_j \in N_t \wedge obs[p_j] \neq \bot \}$
\If {$|O| \geq 2f_t+1 $} \label{Li:fpthreshold}
\State \send $\Vprop(O,\median{O},r)_i$ to all nodes in $N_t$ \label{Li:fpsendmedian}
\EndIf
\UPONEND

\UPONSTART{receiving $\Vprop(O,v,r)_j$ \textbf{and} $p_j \in \aggset$}
\If {$Validate(\Vprop(O,v,r)_j)=\true$} \label{Li:fpvalidatemedian}
\State \send signed vote  $\Votevp(O,v,r)_i$ to $p_j$ \label{Li:fpsignedmedian}
\EndIf
\UPONEND

\UPONSTART{receiving  $\Votevp(O,\median{O},r)_j$ from $p_j$ in $N_t$ \textbf{and} $p_i \in \aggset$} 
\If{\quorum is formed on $\Votevp(O,\median{O},r)$} \Comment{Quorum with $2f_t+1$ votes}
\State $\postsmr{\Vpost(O,\median{O},r,\quorum)}$   \label{Li:fppostsmr}
\EndIf
\UPONEND

\UPONSTART{witnessing the {\bf first} $\Vpost(O,v,r,\quorum)$ on \smr} \label{Li:fpwitnessvpost}
\State $\svalue_r \gets v$ \label{Li:fpsvaluemedian}
\State \Return
\UPONEND
\end{algorithmic}
\end{algorithm}

The fallback protocol where the entire tribe $N_t$ participates is shown in \cref{Alg:fallbackprotocol}. Note that we assume that at most $f_t \leq \frac{|N_t|-1}{3}$ many nodes in the tribe may exhibit Byzantine behaviour.
\begin{enumerate}
	\item All the nodes gather observations from their assigned data sources and then send  corresponding medians to the aggregators (\crefrange{Li:fpinitstart}{Li:fpinitend}). Note that even the clan nodes that triggered the fallback gather fresh data along with the rest of the tribe. 
	\item When an aggregator receives a $\Value$ from a node, it checks whether it has received inputs from $2f_t+1$  nodes (\cref{Li:fpthreshold}). Since at most $f_t$ out of total $|N_t|=3f_t+1$ nodes may be Byzantine, by gathering inputs from $2f_t+1$ votes, the aggregator is assured that within this set of inputs, there is an honest majority. From these inputs, it calculates its median and sends $\Vprop(O,\median{O},r)$  to all the nodes in the tribe $N_t$ for $\median{O}$ to be considered as $\svalue_r$ for the current round $r$ (\cref{Li:fpsendmedian}).
	\item A node receiving a proposal $\Vprop(O,v,r)$ from an aggregator validates that (i) observations in $O$ are signed by nodes in $N_t$, (ii) the value $v$ is indeed the median of $O$, and (iii) the message is indeed from an aggregator in \aggset (\cref{Li:fpvalidatemedian}). It sends back its vote $\Votevp(O,\median{O},r)$ to the same aggregator  (\cref{Li:fpsignedmedian}).
	\item An aggregator, upon receiving $2f_t+1$  votes, prepares the quorum certificate \quorum and  posts $\Vpost(O,\median{O},r,\quorum)$  on \smr  (\cref{Li:fppostsmr}). 
	\item The nodes conclude round $r$ after reaching consensus on $v$ as $\svalue_r$  when they witness $\Vpost(O,v,r,\quorum)$ on \smr (\cref{Li:fpsvaluemedian}).
	\end{enumerate}

It is evident from~\cref{The:median} that the $\svalue_r$ reported by the fallback protocol will always be inside the interval $[\hmin,\hmax]$. Due to multiple aggregators, it is possible that an $\FTpost$ is posted on \smr, but a $\Vpost$ by another aggregator emerging from ~\cref{Alg:defaultfallback} is delayed in  reaching the \smr. In this case, the nodes would switch to the fallback protocol and then witness the delayed $\Vpost$ message. In that case, the $\svalue$ would be within $(\hmin-\agreedist,\hmax+\agreedist)$ as defined in \cref{prop:doraccvalidity}. The nodes would only consume the $\svalue$ from the first $\Vpost$ message on \smr for any given round.

\begin{theorem}[Fallback Termination] \label{The:fallbacktermination}  \cref{Alg:fallbackprotocol} eventually terminates for all honest nodes if (i) $N_t$ has an honest super majority and (ii) $\aggset$ has at least one honest aggregator.
\end{theorem}
\begin{proof}
\ifasyncdstonode
$GetDataFeed()$ terminates as per \cref{The:dfterminate}.
\else
	$GetDataFeed()$ terminates within $\dstimebound$ as per the design. 
 \fi
 All $2f_t+1$ honest nodes would send their corresponding median values to all the aggregators. Since there is at least
	one honest aggregator, it would eventually receive  values from $2f_t+1$ nodes. It is possible that in the subset $O$ considered by the aggregator, there are some Byzantine observations.
	However, the aggregator would send $\Vprop(O,\median{O},r)$ to all the nodes in $N_t$. An honest node can validate the aggregator's proposal, even if its own value is not in $O$.
	Therefore, an honest aggregator would eventually receive $2f_t+1$  votes on its proposal. This would be posted on \smr by the aggregator, ensuring at least one valid proposal on \smr. Therefore, all the honest nodes would eventually witness
	some valid   $\Vpost(O,v,r,\quorum)$ on \smr and would reach a consensus on $\svalue_r$ to be $v$.
\end{proof}

\begin{theorem}[\dora-CC-Fallback Termination] \label{The:doraccfallbacktermination} \cref{Alg:defaultfallback} eventually terminates for all honest nodes if (i) $N_c$ has honest majority, (ii) $N_t$ has honest super majority, and (iii) $\aggset$ has at least one honest aggregator.
\end{theorem}
\begin{proof}
\ifasyncdstonode
All the honest nodes would finish gathering their data feed as per \cref{The:dfterminate}.
	\else
 All the honest nodes would finish gathering their data feed within $\dstimebound$.
\fi
	We will prove the termination of round $r$ by all the honest nodes. 

	An honest node would conclude a round and move to the next round only upon witnessing a $\Vpost(Obs,v,r,\quorum)$ for some set of observations $Obs$. There are multiple scenarios to be considered here.

	\begin{enumerate}
		\item If at least one of the aggregators is able to form a coherent cluster and post a value via $\Vpost$ message on \smr as the first message for round $r$, then all the honest nodes would witness it and conclude round $r$. (\cref{Li:dpvpostmusmr}).
		\item If the first message for round $r$ that appears on \smr is an  $\FTpost$ message (\cref{Li:dpvwitnessfb}). In this case, all of the honest nodes would switch to the fallback protocol (\cref{Alg:fallbackprotocol}). It is possible that some other aggregator in $\aggset$  is able to post a value via a $\Vpost$ message on \smr after the $\FTpost$ message mentioned earlier. Even in this case, all the honest nodes would witness this $\Vpost$ in \cref{Alg:fallbackprotocol} (\cref{Li:fpwitnessvpost}) and conclude round $r$.  If no $\Vpost$ is posted on \smr from \cref{Alg:defaultfallback} after an  $\FTpost$, since all the honest nodes would have switched to the fallback protocol, we are guaranteed termination due to \cref{The:fallbacktermination}.
	\end{enumerate}

	Therefore, we need to prove that from \cref{Alg:defaultfallback} either $\Vpost$ or $\FTpost$  will definitely be posted.
	\begin{enumerate}
		\item If an honest aggregator is able to form a coherence cluster, due to an honest majority in the clan, it would be able to get $f_c+1$ votes approving the cluster and its mean and would be able to send a $\Vpost$ to \smr.
		\item If no aggregator is able to form a coherent cluster, all the honest nodes would timeout on their $\fallbacktimer$ eventually. Therefore, there will be $f_c+1$ $\Voteft$ votes will be sent to all the aggregators  eventually and some honest aggregator would post $\FTpost$ on \smr.
	\end{enumerate}

\end{proof}

SMR safety (\cref{prop:smrsafety}) ensures that all the nodes eventually observe messages in the exact same total order. The agreement property (\cref{prop:doraagreement}) follows directly from this SMR safety property.
Notice that for $n>2f$ with $(n,n-f)$ threshold signature setup, an aggregator may create multiple signed quorums; however, the SMR safety again helps as the nodes pick the first published quorum.

\subsubsection{Analysis of Communication Complexity \label{Se:dorafallbackcomm}}

Let us analyze how many messages and bits need to be transmitted by \cref{Alg:defaultfallback} and \cref{Alg:fallbackprotocol}. We shall use  $n_c=|N_c|$ to denote the size of the clan,  $n_a = |\aggset|$ to denote the size of the family of aggregators, and $n_t=|N_t|$ to denote the size of the tribe.

To obtain data from the data sources, we need $(2f_d+1)n_c$ messages. Nodes would then send $n_an_c$ $\Value$ messages to aggregators. After aggregation, the aggregators would send $n_cn_a$ $\Vprop$ messages. The nodes would generate $n_cn_a$ $\Votevp$ messages to be sent to aggregators which would subsequently send at most $n_a$ messages to the \smr. This would result in the total number of messages being transmitted by \cref{Alg:defaultfallback} to $(2f_d+1)n_c + 3n_cn_a + n_a$ when a coherent cluster is formed. Considering $n_cn_a$ to be the dominating term, the message complexity would be $O(n_cn_a)$ in the optimistic scenario. Let the total number of variables for which we need to run \dora be denoted as $n_{\commodity}$. If we run \dora for different $\commodity$ independently, then the message complexity would increase to $O(n_cn_an_{\commodity})$. However, one can batch \dora messages for different \commodity together in which case the message complexity would again reduce to $O(n_cn_a)$.

In the case when $\fallbacktimer$ times out, the nodes directly send votes to initiate the fallback. This would require $O(n_cn_a)$ messages at the most. Though each aggregator may try to post $\FTpost$ to \smr, only one can succeed and the other will be discarded as a duplicate by \smr. However, it should be counted since each aggregator may transmit this message, thus requiring a total of $O(n_a)$ messages. 
When the nodes switch to the fallback protocol, the data feed now requires $O((2f_d+1)n_t)$ messages. There would be $O(n_tn_a)$ $\Value$ messages by the nodes, $O(n_tn_a)$ $\Vprop$ messages by the aggregators and $O(n_tn_a)$ $\Votevp$ messages by the nodes. The aggregators would at most transmit $O(n_a)$ messages to the \smr.  Therefore, the total number of messages in a given round is $O(n_tn_a)$, since $n_t > n_c$.

	For analyzing complexity in terms of the number of bits transmitted, we need to take into account the size of each message. Let the length of the data-feed message and  $\Value$ message and hash of any messages be upper-bounded by $k$. We assume that within this length we can store all the information needed for the protocol such as a round identifier, a node identifier, a data-source identifier, a \commodity value etc. The data feed stage would require $O(k(2f_d+1)n_c)$ bits. Let us assume $\lambda$ to be the length of a signature. There would be $O((k+\lambda)n_cn_a)$ bits required for the nodes to send $\Value$ messages to aggregators. Since the aggregator sends back not only the $\mu$ but also the set of original $\Value$ messages forming a coherent cluster, the number of bits required for $\Vprop$ messages would be $O(((k+\lambda)n_c + k+ \lambda)n_c  n_a)$. The nodes send their votes on a value proposal, thus requiring $O((k+\lambda) n_cn_a)$ bits for $\Votevp$ messages. We assume that the nodes would sign the hash of the $\Vprop$ message received earlier and return it as an approval.   The aggregators then form a \quorum and send it to the \smr, which would require $O((k+\lambda n_c)n_a)$ bits of transmission since each \quorum would have its size in proportion to $n_c$. Therefore, the communication complexity in bits would be $O((k+\lambda)n^2_cn_a)$ for \dora-CC. For $n_{\commodity}$ many variables it would be $O((k+\lambda)n^2_cn_an_{\commodity})$.

	In case the fallback happens, the communication complexity in terms of bits would be similar to the one described above, but now the messages would be transmitted at a tribe level. Thus the number of bits required for a given round would be $O((k+\lambda)n^2_tn_a)$ for a single \commodity and $O((k+\lambda)n^2_tn_an_{\commodity})$ for $n_{\commodity}$ variables.

\subsubsection{Error analysis}
We define error of the protocol as the deviation that a presence of Byzantine nodes may cause in the output of the protocol.
When the number of Byzantine observations is $f$ and we consider $2f+1$ observations, then we know that as per \cref{The:median} the median value will fall within the bounds defined by $\hmin$ and $\hmax$. However, this is not the case when we only consider $f+1$ observations. A Byzantine aggregator could find $\hmax$ and insert $f_c$ Byzantine observations with value $\hmax+\agreedist$ to form a cluster. It is possible that all the other $f_c$ honest observations have the value $\hmin$. Had only the honest observations been considered to form a cluster, the value of $\mu$ would have been $\frac{f_c\hmin+\hmax}{f_c+1}$ but instead, we would end up with $\mu = \frac{\hmax+f_c(\hmax+\agreedist)}{f_c+1}$. 
Therefore, we would have an error upper bound of $\|\hmin-\hmax\|_1+\agreedist$ when \cref{Alg:defaultfallback} emits $\svalue_r$ via the \dora-CC protocol.

In the case of \cref{Alg:defaultfallback}, it may be possible that $\agree{\hmin}{\hmax}$ does not hold and an honest aggregator would have proposed a fallback. However, a Byzantine aggregator may be successful in preventing the fallback by forming a cluster of one honest observation $\hmax$ with $f_c$ Byzantine observations with value $\hmax+\agreedist$. Had the protocol switched to the fallback protocol, the smallest value of the median produced by \cref{Alg:fallbackprotocol} would have been $\hmin$. Therefore, even in this case, the upper bound for the error would be  $\|\hmin-\hmax\|_1+\agreedist$. The arguments with respect to Byzantine observations forming a cluster with $\hmin$ would be symmetric and  do not result in any change in the error upper bound.

\begin{theorem}[Least error upper bound] \label{The:minerror} A protocol for agreeing on a \svalue, where $3f_t+1$ nodes participate out of which $f_t$ of these nodes could be Byzantine, and a median of values from non-deterministically chosen\footnote{The non-determinism choice is introduced due to the non-determinism in network delay} $2f_t+1$ nodes is proposed as the \svalue, would have an error upper bound of at least $\|\hmin-\hmax\|$.
\end{theorem}
\begin{proof}
	Out of the total $3f_t+1$ nodes, $f_t$ could be Byzantine. Out of $2f_t+1$ honest observations, let $2f_t$ of them have the value $\hmin$ with one honest observation having the value $\hmax$. Let all the Byzantine observations have the value $\hmax+c$, where $c>0$. Out of these $3f_t+1$ observations, if the median is calculated from only $2f_t+1$ honest observations, the median would be $\hmin$. Instead, if the median is calculated from $2f_t+1$ values where $f_t$ honest values are $\hmin$, one honest value is $\hmax$ and all the Byzantine observations have the value $\hmax+c$ then the median would be $\hmax$.
	Therefore, for any such protocol, the largest possible error can not be less than $\|\hmin - \hmax\|_1$.
\end{proof}

Note that \cref{The:minerror} would hold for any protocol, that computes a median of only $2f_t+1$ values out of the total $3f_t+1$ values that may be available.

It is evident from \cref{The:minerror} that \cref{Alg:defaultfallback} may increase the error upper bound by only $\agreedist$.

\subsubsection{Error compounding}

Since it is possible that Byzantine actors may introduce an error of $\errorbound$, one can wonder if a Byzantine aggregator across several successive rounds can compound the error. This is not possible
because an $\svalue_r$ is bounded by $\hmin$ and $\hmax$ of honest observations from the round $r$ only. Therefore, the protocols in  \cref{Alg:defaultfallback} or \cref{Alg:fallbackprotocol} do not allow the error to be compounded in successive rounds.

\section{Analysis \label{Se:empirical}}
\subsection{Empirical Analysis of Agreement Distance \label{Se:agreedist}}
We would describe some empirical analysis and simulations of our protocol done on real-world data.

We obtained data for \BTC price in \USD from $7$ different exchanges, namely, \binance, \coinbase, \cryptocom, \ftx, \huobi, \okcoin, and \okex from \adiststartdate to \simenddate. We divided this data into two parts. The first part consists of the data from \adiststartdate to \adistenddate which was used in determining various values for \agreedist. The second part consists of the data from \simstartdate to \simenddate including the turbulent \ftx collapse period, which was used to simulate the protocol with various values of \agreedist.

The data from \adiststartdate to \adistenddate was divided into (i) $30$ second windows, and (ii) $60$ second windows. For each window, the median and the mean were calculated from all the values/observations available within that window. We observed that the mean of these means and the mean of these medians are less than $\$0.02$ away from one another.
For $30$ second and $60$ second windows, the mean of means was around $\$19605$ and $\$19606$ respectively. Therefore, we used $\$19605.5$ as the representative price of \BTC for the duration of \adiststartdate and \adistenddate.

 The simulations were done with \agreedist set to various values from $0.02\%$ to $0.55\%$  of $\$19605.5$, the mean value of \BTC calculated as described above.

For simulation also, we used $30$ and $60$ seconds as two different values for $\dstimebound$. The simulation data from \simstartdate and \simenddate was divided into (i) $30$ second windows, and (ii) $60$ second windows. Therefore, every node  obtained its data from data sources from the same window for a given round of oracle agreement. We assume that the nodes of the oracle network have clock-drift within a few hundred milliseconds.~\cite{FBEngineering}
 We simulated the behavior of $7$ nodes. Each node was randomly assigned $5$ out of the $7$ exchanges.  Data from each window was used to simulate one round of agreement. 

 Availability of data from various exchanges is shown in \cref{Tab:exnullvalue}. Note that $4$ out of $7$ exchanges did not provide any value more than $50\%$ of the time for the simulation duration of \simstartdate to \simenddate. Therefore, our design choice of using multiple data sources per node is justified to improve the reliability of our protocol. Notice that \okcoin produces a value almost at every minute, therefore, while for $30$ second windows, it provides values for almost every alternate window, for $60$ second window its availability is very high.

For every round, every node would compute the median from the values it obtained from its $5$ assigned exchanges. If an exchange had multiple values within the given window, only the latest would be considered by the node. Once every node has computed its median, these medians would be sorted to see if any $4$ of the $7$ nodes form a coherent cluster for a given value of \agreedist. 
\begin{figure}
\begin{tikzpicture}[scale=0.9]
        \begin{axis}[title=Cluster formation vs agreement distance,
        xlabel={Agreement distance},
        ylabel={Cluster formation percentage}, minor y tick num=10, grid,  legend style={anchor=north east, at={(0.9,0.2)}}]
\addplot [blue] table [col sep=comma,x index=0, y index=1]  {ccformation.csv}; 
\addlegendentry{30 sec window};
\addplot [red,dashed] table [col sep=comma,x index=0, y index=2]  {ccformation.csv};
\addlegendentry{60 sec window};
\end{axis}
\end{tikzpicture}
\caption{Percentage of times coherent cluster was formed as agreement distance increases. Agreement distance is specified as percentage of $19605$, the average \BTC price  }
\label{Fi:ccformation}
\end{figure}
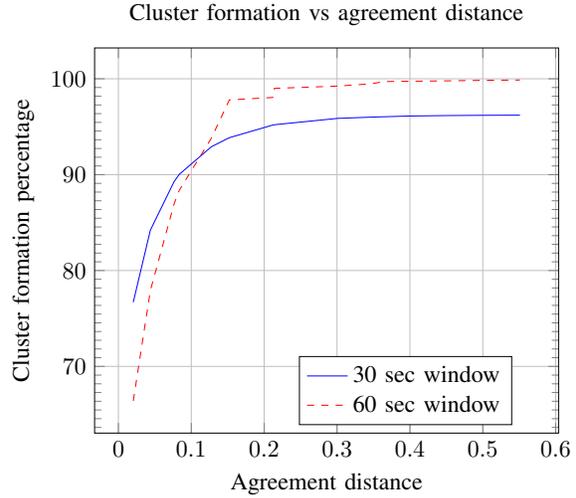
\cref{Fi:ccformation} shows the results from our simulation. As the size of the observation window increases, one would naturally expect to see more deviation in values among values of various exchanges. Therefore, one would expect that for a given \agreedist, the nodes would be able to form coherent clusters more often for $30$ second windows as compared to $60$ second windows. With \agreedist of $\$16$, coherent cluster is achieved $90\%$ and $88\%$ for $30$ and $60$ seconds window respectively. However, note the drastic increase in data availability from \okcoin from $30$ second to $60$ second observation period. This has contributed to the increased percentage of coherent cluster formation for a $60$ second window for larger \agreedist. For example with \agreedist being $\$53$, cluster formation is achieved for $99\%$ of the time for $60$ second window.

Note that the choice of \agreedist has a bearing on both the safety and the performance of the protocol. Smaller values of \agreedist would result in the protocol having to fall back more often, whereas higher values of \agreedist potentially  allow higher deviation of \svalue from the ideal representative value of the mean of honest values. 

If the value of a \commodity increases or decreases significantly, for safety and performance reasons, \agreedist should be increased and decreased accordingly.

\begin{table}
\centering
\caption{Percentage of times various exchanges did not produce any value for a time window.}
\label{Tab:exnullvalue}
\pgfplotstabletypeset[col sep=comma, row sep=newline, display columns/0/.style={string type, column type=|c}, display columns/1/.style={column type=|p{2cm}|}, display columns/2/.style={column type=p{2cm}|}, after row=\hline, every head row/.style={before row=\hline}]
{
Exchange, Null value percentage for 30 sec, Null value percentage for 60 sec
\binance,51.653,52.66
\cryptocom,78.881,75.712
\coinbase, 97.018, 96.306
\ftx,73.873,74.116
\huobi,44.447,45.509
\okcoin,46.659,0.495
\okex,26.837,28.334
}

\end{table}

Increasing \fallbacktimer may have have two effects. When prices are not fluctuating too much, a large \fallbacktimer allows the round to complete and \svalue to appear on \smr. However, large value of \fallbacktimer increases round completion time when coherent cluster can not be formed and a fallback is required. Reducing \fallbacktimer to a very small value results in much higher percentage of rounds resulting in a fallback, thus, degrading the performance. Therefore, \fallbacktimer should be set after measuring average time to complete a round of \doracc.

\subsection{Theoretical Analysis of Probabilistic Safety \label{Se:probsafety}}

  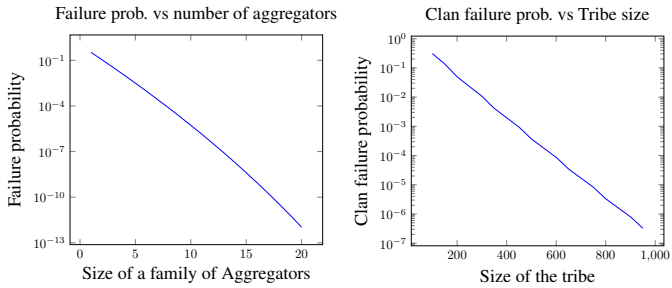
\begin{figure}
   \centering
     \begin{subfigure}[b]{0.48\columnwidth}
         \centering
\begin{tikzpicture}[scale=.49]
\begin{semilogyaxis}[title={\Large Failure prob. vs number of aggregators},
	xlabel={\Large Size of a family of  Aggregators},
	ylabel={\Large Failure probability}]
\addplot [blue] table [col sep=comma]  {aggfamilysize.csv};
\end{semilogyaxis}
\end{tikzpicture}
\caption{Probability of having no honest aggregators in the family of aggregators as the size of the family increases. Here total nodes are $100$ and Byzantine nodes are $33$}
\label{Fi:aggfamilyprob}
     \end{subfigure}
     \hfill
     \begin{subfigure}[b]{0.48\columnwidth}
         \centering
\begin{tikzpicture}[scale=.49]
\begin{semilogyaxis}[title= {\Large Clan failure prob. vs Tribe size},
	xlabel={\Large Size of the tribe},
	ylabel={\Large Clan failure probability},]
\addplot [blue] table [col sep=comma, x index=0, y index=3,]  {clanfailprob.csv};
\end{semilogyaxis}
\end{tikzpicture}
\caption{Probability of having at least one clan with a Byzantine majority as the size of the tribe increases. Here number of clans are $5$ and Byzantine nodes are $33\%$}
\label{Fi:clanfailprob}
  \end{subfigure}
        \caption{Probabilistic Safety Analysis}
\label{fig:ProbabilisticSafety}
\end{figure}

We mention in \cref{Se:protocol} that we employ multiple aggregators. \cref{Fi:aggfamilyprob} shows a logarithmic plot on how the probability of having an entire family consisting of Byzantine aggregators reduces as we increase the size of the family of aggregators. Since $1$ in $3$ nodes in the tribe is Byzantine, it is evident that as we increase the size of the family of aggregators, the probability of not having a single honest aggregator drops exponentially fast.
\cref{Fi:clanfailprob} shows how the probability of having at least one clan with a Byzantine majority changes as we increase the size of the tribe. In the logarithmic plot,  we can observe that increasing the size of the tribe to a few hundred nodes would ensure  that none of the $5$ clans randomly drawn from the tribe would have a Byzantine majority with a very high probability. To fully exploit the ability to tolerate up to $49\%$ Byzantine nodes, we propose to employ probabilistic safety guarantees. \cref{Fi:clanfailprob} provides guidelines on how many total nodes would be required to achieve the safety property with a very high probability.

Note that, if the adversary is fully adaptive and can compromise nodes after families and clans are formed, the above probabilistic safety analysis is not applicable. We expect clans and families to be updated regularly in practice. 

\section{Related Work}
 In this section, we compare our approach to relevant prominent industrial as well as academic solutions to the oracle problem.

As we describe earlier, Chainlink's Off-Chain Reporting Protocol (OCR)~\cite{ocr} is the most employed oracle solution today. Under partial synchrony, OCR requires Byzantine nodes to be less than $\frac{1}{3}$  fraction of the total population. 
%
Our unique ideas of defining agreement between nodes based on {\em agreement distance} and better leveraging \smr allow our protocol to tolerate up to $\frac{1}{2}$ Byzantine nodes. 

Note that the OCR approach inherently will not be able to tolerate more than $\frac{1}{3}$ Byzantine nodes
even if we make the coherent cluster assumption on its input and update the OCR protocol: 
this approach uses the blockchain (smart contract) purely for broadcasting a report/result and not for ordering and subsequently preventing equivocation. As a result,
it will require more than $\frac{2}{3}$ honest nodes for broadcast itself~\cite{Non-equivocation}
While the OCR approach is indeed more self-sufficient, we argue that if one has an access to \smr (even via a smart contract), there is an advantage and merit in leveraging it in order to be able to tolerate a higher fraction of Byzantine nodes.

Moreover, our DORA protocol is asynchronous by design; thus, if the underlying blockchain is secure under asynchrony, our DORA protocol also becomes secure under the asynchronous setting. Our timer $\fallbacktimer$ does not depend on the transmission time, the adversary can only trigger the fallback protocol using asynchrony.




In the Pyth network~\cite{pyth},
sources directly post their signed values to a smart contract.
While we do not have to worry about Byzantine oracle nodes in that case and the oracle problem becomes purely an availability problem there,
most data sources do not employ the digital signature infrastructure today. Even if some data sources start to employ data authentication primitives in the future, there will be compatibility issues as different blockchains employ different signature mechanisms.



Provable~\cite{provable} and Town Crier \cite{TownCrier} are earlier approaches that utilize trusted-execution environments (TEEs) 
to furnish data to the blockchains. However, TEEs have proven to be notoriously hard to secure and the security provided by TEEs has been broken by a wide variety of attacks.
As a result, TEE-based oracles did not get enough traction yet.

As TLS becomes the de facto standard for (secure) communication  over the internet, TLS-N~\cite{TLSN}, DECO~\cite{DECO}, and recently ROSEN~\cite{ROSEN} allow oracle nodes to prove facts about their TLS sessions to data sources. However, there remain several challenges.
TLS-N and ROSEN enable a data source to
sign its TLS session with the oracle node; however, they need non-trivial updates to the TLS protocol and thus suffer from high adoption costs.
 DECO ensures that the only TLS client has to be modified but no server-side modifications are required, which makes it better suited as only the oracle nodes' TLS code has to be updated. 
However, the proposed custom three-party handshake mechanism allows the creation of proofs that are acceptable only by
a pre-specified, actively participating verifier,
which may not be applicable to most blockchain nodes. 

We refer the reader to the blockchain oracle surveys \cite{ACMSurvey, 9086815} for  discussion about several other informal blockchain oracle designs.

\section{Conclusion and Future work}
We present a novel distributed oracle agreement protocol that allows the oracle network to function with only $2f+1$ nodes when the prices of a commodity are not fluctuating wildly, by leveraging \smr as an ordering primitive and updating the notion of agreement among nodes. 

We have shown that data sources do pose a data availability risk and therefore it is wise to mandate nodes to gather data from multiple data sources. We have also shown the trade-off that the agreement distance parameter offers in terms of safety and performance.

We show that it is possible to build safe and efficient oracle networks with high probability with appropriate choices on the size of the network and the size of the family of aggregators.

In the future, it can be interesting to explore how \agreedist can be adjusted in an algorithmic fashion without allowing any room for the compounding of errors by a malicious entity.
\section*{Acknowledgements}
We thank anonymous reviewers and Akhil Bandarupalli for their helpful comments. 
We thank Luying Yeo for assisting with creative figures.
\bibliographystyle{IEEEtran}
\bibliography{references}
\ifextendedversion
\else
\appendices
\section{Extensions}
\subsection{Assigning weights to the data sources}

In \cref{Alg:datafeed}, we mandate the nodes listen to $2f_d+1$ data sources, since $f_d$ of these data sources can turn Byzantine. Each node computes the  median of all the observations it makes from these $2f_d+1$ data sources. This algorithm, however, treats every data source equally. Different data sources, however, may have different reliability guarantees, trading volumes, downtime statistics and security practices, etc. It makes sense that when we look at these aspects, different data sources are treated differently. For example, a data source that provides better reliability guarantees and follows the best security practices should be given more weight when compared to some other data sources which may be inferior in these parameters. 

To accommodate the differing traits of different data sources, we can define a weight function $ w : DS \rightarrow \Nat^+$, which assigns different positive integer weights to different data sources. The weight can be determined by considering parameters like historical reliability, volume/scale, and security practices followed by the data source etc., and a suitable value can be assigned through data governance. It is important to note that while we assign weight to the data sources, a data source crash or malfunction would affect any function that takes its inputs in proportion to its weight. Therefore, we now must discuss the total weight of the data sources that may turn Byzantine. In this modified setting, we can think of $f_d$ as the sum of the weights of the data sources that can turn Byzantine. To accommodate this change, \cref{Alg:datafeed} can be modified where every node listens to data sources such that their total weight is $2f_d+1$. In this case, whenever a value $v$ is received from a data-source $ds$, we would store as many as $w(ds)$ copies in the multi-set $obs$ (\cref{Li:dfgetvalue}). The size of $Obs$ would still be between $f_d+1$ and $2f_d+1$, the lower bound on the number of observations in $obs$ from honest data sources would still be $f_d+1$ and the number of observations from Byzantine data sources would be upper-bounded by $f_d$. \cref{The:median} would still hold in this case. The rest of the protocol for \dora would remain the same.

\fi

\end{document}